\RequirePackage{xcolor}
\documentclass[9pt,shortpaper,twoside,web]{ieeecolor}
\usepackage{generic}
\usepackage{url}
\usepackage{booktabs}

\usepackage{comment}
\usepackage{cite}
\usepackage{amssymb}
\usepackage{graphicx} % for pdf, bitmapped graphics files
\usepackage{epstopdf}
\usepackage{caption}
\usepackage{algpseudocode}
\usepackage{algorithm}
\usepackage{balance}
\usepackage[fancythm,fancybb]{jphmacros2e} 
\usepackage{amsfonts}
\usepackage{mathrsfs}
\usepackage{amsmath}
\usepackage{mathtools}
 \usepackage[norefs,nocites]{refcheck}
\let\labelindent\relax
\usepackage{enumitem}
\usepackage{titlesec}
\usepackage{subcaption}

\newcommand{\tr}{\mathrm{T}}

\newcommand{\norm}[1]{\left\lVert#1\right\rVert}

\makeatletter
\g@addto@macro\normalsize{%
\setlength\abovedisplayskip{4.5pt}            %Equation space
\setlength\belowdisplayskip{4.5pt}
\setlength\abovedisplayshortskip{4.5pt}
\setlength\belowdisplayshortskip{4.5pt}
}
\makeatother

\titlespacing{\section}{0pt}{0.5\baselineskip}{0.5\baselineskip}
\titlespacing{\subsection}{0pt}{0.4\baselineskip}{0.4\baselineskip}
\begin{document}
%\IEEEoverridecommandlockouts
	\allowdisplaybreaks

\title{Adversarial Observability and Performance Trade-offs in Optimal Control}
\author{Filippos Fotiadis$^1$,~\IEEEmembership{Member,~IEEE}, Ufuk Topcu$^1$,~\IEEEmembership{Fellow,~IEEE}   
\thanks{$^1$ F. Fotiadis and U. Topcu are with the Oden Institute for Computational Engineering \& Sciences, University of Texas at Austin, Austin, TX, USA. Email:
\{ffotiadis, utopcu\}@utexas.edu.}
\thanks{
This work was supported in part by ARO under grant No. W911NF-23-1-0317, by ONR under grant No. N00014-24-1-2097, and by NASA under grant 80NSSC24M0070.}
}
\maketitle

\begin{abstract}
     We develop a feedback controller that minimizes the observability of a set of adversarial sensors of a linear system, while adhering to strict closed-loop performance constraints. We quantify the effectiveness of adversarial sensors using the trace of their observability Gramian and its inverse, capturing both average observability and the least observable state directions of the system. We derive theoretical lower bounds on these metrics under performance constraints, characterizing the fundamental limits of observability reduction as a function of the performance trade-off.  Finally, we show that the performance-constrained optimization of the Gramian's trace can be formulated as a one-shot semidefinite program, while we address the optimization of its inverse through sequential semidefinite programming.   Aircraft and satellite simulations demonstrate that the proposed controllers reduce adversarial observability while maintaining near-optimal performance. 
\end{abstract}
\vspace{-2mm}

\vspace{-2mm}
\section{Introduction}
\vspace{-2mm}

Autonomous systems often face adversaries who want to predict their current state or intent by deploying their own individual sensors.  For example, in aviation, external observers can use radar observations to reconstruct an aircraft's state, identify its intent, and predict its future trajectory \cite{yepes2007new}.  In space, external observers can determine a satellite's active control mode using only unresolved optical images and without needing onboard telemetry data \cite{coder2018inferring}. To ensure security, such systems must thus be able to maneuver in a way that minimizes the observability of potentially malicious observers.

A potential solution is moving target defense -- a security method that aims to make the system's evolution unpredictable and thus disrupts adversarial observability  \cite{kanellopoulos2019moving, sandbergmoving, sinopolimoving, seta}. The underlying concept is that, to impede an adversary from performing system reconnaissance or from predicting future trajectories, one should constantly change the system's control mode. This induces a randomization in the evolution of the system that renders future tasks of the system difficult to identify.
However, switching between different system modes can create discontinuities and transients in control execution, potentially degrading performance and compromising safety. Moreover, when the system has to commit to meeting a single objective, employing different control modes may not always be feasible. These observations highlight the need for observability-reduction mechanisms that do not rely on mode switching, but instead operate within a single, performance-constrained feedback architecture.

%However, switching between different system modes can create discontinuities and transients in control execution, potentially degrading performance and compromising safety. Moreover, when the system has to commit to meeting a single objective, employing different control modes may not always be feasible.

To obstruct the state reconstruction capability of adversarial sensors,
several studies have investigated the design of a single controller that
reduces observability \cite{Zhang2023, Zhang2023TAC, Roy2019,
al2022observability}. Related work has also studied the observability
radius as a metric of distance to unobservability \cite{bianchin}, network
design for controllability metrics \cite{becker}, and active
defense strategies that mislead adversarial observers by modifying their
available signals \cite{shaaban}. These approaches, however,
either treat observability reduction as a binary constraint, focus on
network topology rather than controller synthesis, or assume knowledge of
a specific adversarial estimator. Moreover, they do not explicitly
quantify the trade-off between closed-loop performance and observability
reduction. This motivates feedback controllers that enable continuous,
tunable trade-offs between performance and adversarial
observability. \hspace{-2mm}

 We consider a setup where the system's operator, who has full state feedback, wants to design a linear controller to balance two conflicting objectives: i)  to minimize the observability of a set of sensors which an adversary might be using to observe the system; and ii) to stabilize the system to the origin with optimal performance. We cast this dual-objective control problem as a constrained optimization of a metric of unobservability of the adversary's sensors, subject to a trade-off hard constraint on the distance of closed-loop performance from optimality. In contrast to the aforementioned studies that treat unobservability as a binary property, we capture observability continuously, drawing from the sensor selection literature and leveraging observability Gramians \cite{summers2015submodularity, pasqualetti2014controllability, bopardikar, fotiadis2024input, manohar}.

 Gramians are a prime tool for capturing the controllability and observability of linear systems in a continuous manner that goes beyond rank tests. Observability Gramians, in particular, must be inverted to perform least squares state estimation from partial observations \cite{carnevale2013nonlinear}, and hence their spectrum directly affects the quality of the resulting state estimate. For this reason, there is a rich body of literature that correlates how one should select the output nodes of a system with the underlying observability Gramian of those nodes \cite{summers2015submodularity, pasqualetti2014controllability, bopardikar, fotiadis2024input, manohar}.  Here, we explore this connection in the context of adversarial observability minimization. Unlike sensor selection, which chooses sensor configurations to maximize observability, our objective is to design the controller to actively reduce the adversary's ability to observe the system.

 We specifically quantify adversarial observability using the trace of the observability Gramian and its inverse, which respectively capture average observability and weakly observable directions. For both metrics, we derive theoretical lower bounds as functions of the allowable performance degradation. We show that smaller performance trade-offs produce a sharper decrease in the trace of the observability Gramian, whereas larger trade-offs produce a steeper increase in the trace of its inverse. Finally, we solve the first optimization problem exactly via semidefinite programming (SDP) and the second approximately using a sequential SDP algorithm.

%We specifically quantify the observability of adversarial sensors by adopting the trace of the observability Gramian and its inverse. The former quantifies adversarial observability on average, across all observable directions, while the latter is skewed towards the less observable directions. For both of these metrics, we establish theoretical lower bounds on their optimal values as a function of the trade-off distance from optimal performance. We show that smaller performance trade-offs generally yield a sharper decrease in the trace of the observability Gramian, whereas larger performance trade-offs lead to a steeper increase in the trace of the Gramian's inverse. Finally, we solve the constrained optimization of these metrics i) precisely for the trace of the Gramian, by proving it boils down to the solution of a semidefinite program (SDP); and ii) approximately for the trace of the Gramian inverse, through a sequential SDP algorithm that we obtain by applying the convex-concave procedure.

\textit{Notation:} For any matrix $X$, $\norm{X}$ denotes its operator norm.  For a square matrix $X$, $\mathrm{tr}(X)$ denotes its trace, and $\underline{\lambda}(X)$,  $\bar{\lambda}(X)$ denote its minimum and maximum eigenvalue when $X$ is symmetric, and $\kappa(X)$ the condition number. We denote $X\succ0$ ($X\succeq0$) if $X$ is positive definite (positive semidefinite), and $X\prec0$ ($X\preceq0)$ if $X$ is negative definite (negative semidefinite). We use $I$ to denote an identity matrix of appropriate dimensions.

\section{Problem Formulation}\label{sec:pr}

Consider, for all $t\ge0$, the continuous-time system
\begin{equation}\label{eq:sys}
\dot{x}(t)=Ax(t)+Bu(t),~x(0)=x_0,~t\ge0,
\end{equation}
where $x(t)\in\mathbb{R}^n$ is the state, $u(t)\in\mathbb{R}^m$ is the control input, and $A\in\mathbb{R}^{n\times n}$, $B\in\mathbb{R}^{n\times m}$ are the system's state and input matrices. 

\subsection{Operator and Adversary Models}

We assume that an adversary has access to the following partial observations of the state vector $x(t)$:
\begin{equation}\nonumber
y(t)=Cx(t).
\end{equation}
We will refer to this partial observation $y(t)\in\mathbb{R}^p$ as the output, and to $C\in\mathbb{R}^{p\times n}$ as the sensing matrix of the adversary. This sensing matrix corresponds to sensors that the adversary is personally employing to observe the state $x(t)$, but could also correspond to existing sensors of \eqref{eq:sys} that the adversary has compromised. 

 We also assume that the operator of system \eqref{eq:sys}, who designs the control input $u(t)$, has access to the full state $x(t)$. In light of this sensing asymmetry, the operator is motivated to design a control law
\begin{equation}\nonumber
u(t)=Kx(t),
\end{equation}
$K\in\mathbb{R}^{m\times n}$, which does not only optimally stabilize system \eqref{eq:sys} to the origin, but also makes $x(t)$ difficult to observe from the output $y(t)$. In other words, the operator wants to design a feedback gain $K$ so that $A+BK$ is Hurwitz and nearly optimal, but also so that $(A+BK,~C)$ is as unobservable as possible.

We consider the following assumptions, which describe the information structure of the operator and adversary.
\begin{assumption}\label{ass:contr_obs}
$(A, B)$ is controllable.
\end{assumption}
\begin{assumption}\label{ass:assym}
The operator can measure the state $x(t)$ and knows the adversary's matrix $C$. The adversary can measure the output $y(t)$ but cannot measure the operator's control input $u(t)$.
\end{assumption}
\begin{remark}
 Assumption \ref{ass:contr_obs} is standard for the operator's objective of optimally stabilizing \eqref{eq:sys} to the origin.
Assumption \ref{ass:assym} is indicative of the information asymmetry between the operator and the adversary, which the operator will use to make the task of observing $x(t)$ from $y(t)$ more difficult for the adversary. Knowledge of $C$ is realistic when the adversary's sensing modality is constrained by physics (e.g., radar observations are determined by the sensing geometry and provide only partial information about the system state, while optical telescopes capture external brightness but not onboard telemetry data \cite{coder2018inferring}).
Finally, while the operator might often not know exactly which partial states the adversary can measure (as indicated by $C$), they are likely to know a superset of them, and hence Assumption \ref{ass:assym} is not restrictive; we formalize this robustness in Proposition \ref{prop:robustness}.
\end{remark}

\subsection{Operator's Control Design Problem}

Per the problem formulation, the operator considers two objectives when designing its control gain $K$: i) choosing $K$ to stabilize \eqref{eq:sys} to the origin optimally; and ii) choosing $K$ to make $(A+BK,~C)$ as less observable as possible. 

We capture the first of the two objectives with the cost function
\begin{equation}\label{eq:LQint}
J_s(K)=\mathbb{E}\left[\int_0^\infty (x^\mathrm{T}(t)Qx(t)+u^\mathrm{T}(t)Ru(t))\mathrm{d}t \right],
\end{equation}
where the expectation is taken over the distribution of the initial condition $x_0\sim\mathcal{N}(0,V)$, and $V\succ0$ is the initial condition's covariance. Using standard linear systems theory, we can then simplify the expression of the cost \eqref{eq:LQint} to
\begin{equation}\nonumber
J_s(K)=\mathrm{tr}(PV),
\end{equation}
where $P\succ0$ is the unique solution of the Lyapunov equation
\begin{equation}\label{eq:LE}
(A+BK)^\mathrm{T}P+P(A+BK)+Q+K^\mathrm{T}RK=0.
\end{equation}

To capture the second objective of minimizing the observability of the pair $(A+BK,~C)$, we resort to costs relating to observability Gramians. Following \cite{pasqualetti2014controllability, summers2015submodularity}, two relevant metrics are
\begin{equation}\label{eq:obsfuns}
\begin{split}
J_{o1}(K)&=\mathrm{tr}(WV),\\J_{o2}(K)&=-\mathrm{tr}(W^{-1}V^{-1}),
\end{split}
\end{equation}
where $W\succeq0$ is the observability Gramian of the closed-loop pair $(A+BK,~C)$, uniquely solving the Lyapunov equation
\begin{equation}\label{eq:lyapW}
(A+BK)^\mathrm{T}W+W(A+BK)+C^\mathrm{T}C=0.
\end{equation}
\begin{remark}
The metrics \eqref{eq:obsfuns} relate to the average energy the system releases through the sensors $C$. Specifically, $\mathbb{E}[\int_0^\infty y^\mathrm{T}(t)y(t)\textrm{d}t]=\mathbb{E}[\int_0^\infty x^\mathrm{T}(t)C^\mathrm{T}Cx(t)\textrm{d}t]=\mathrm{tr}(WV)=J_{o1}(K)$, whereas $J_{o2}(K)$ is motivated by the inequality $\mathrm{tr}(W^{-1}V^{-1})\ge \frac{n^2}{\mathrm{tr}(WV)}$ and the fact that it becomes unbounded when the closed loop is not observable.
\end{remark}

\begin{remark}
The observability Gramian also
governs the conditioning of the adversary's state-reconstruction problem.
Indeed, least-squares reconstruction of $x_0$ from
$y(t)=Ce^{(A+BK)t}x_0+d(t)$ incurs error
$\tilde x_0=W^{-1}\!\int_0^\infty e^{(A+BK)^\mathrm{T}t}C^\mathrm{T}d(t)\,dt$
\cite{carnevale2013nonlinear, chen1984linear}, so the conditioning of the reconstruction problem
is governed by $W^{-1}$ for any disturbance $d(\cdot)$, irrespective of its
temporal correlation. The same Gramian also admits a classical stochastic
interpretation: under white Gaussian measurement noise with variance $\sigma^2$, the Cram\'er--Rao bound gives
$\mathrm{Cov}(\hat x_0)\succeq\sigma^2W(T)^{-1}\succeq\sigma^2W^{-1}$ for any
unbiased estimator $\hat{x}_0$ of $x_0$, where $W(T)$ is the finite-horizon observability Gramian \cite{kunwoo2022observability}.
Thus, reducing $W$ increases the lower bound on the achievable
estimation-error covariance. An alternative measure of unobservability is the
observability radius \cite{bianchin}, which quantifies the smallest perturbation
of $(A,C)$ rendering the system unobservable; however, our objective is not to
certify proximity to rank loss, but to continuously degrade the adversary's
state-reconstruction capability under a hard performance constraint, for which
Gramian-based metrics provide the appropriate smooth measure.
\end{remark}

Notably, to make the pair $(A+BK,~C)$ as unobservable as possible, the operator would typically aim to either maximize $\mathrm{tr}(W^{-1}V^{-1})$ or minimize $\mathrm{tr}(WV)$. The difference between these two metrics lies in the aspect of observability they emphasize. The quantity $\mathrm{tr}(W^{-1}V^{-1})$ is better suited for capturing how close the pair $(A+BK,~C)$ is to being unobservable, as it becomes unbounded (infinite) if the system is not observable. This makes it particularly effective when the objective is to suppress the most weakly observable modes. On the other hand, $\mathrm{tr}(WV)$ is better suited for quantifying observability on average, across all directions of the state space. It remains finite even when $(A+BK,~C)$ is unobservable, and is thus a preferable metric when one does not place much emphasis on complete unobservability, or when a well-defined metric is needed in both observable and unobservable regimes.

In view of the discussion above, to co-optimize the performance of the closed-loop plant matrix $A+BK$ as well as the unobservability of the pair $(A+BK,~C)$, we focus on two optimization problems. The first problem minimizes the first unobservability objective in \eqref{eq:obsfuns}, subject to a constraint that limits the degradation in performance, as quantified by the cost \eqref{eq:LQint} and a trade-off parameter $\lambda > 0$. Specifically, this constraint is $J_s(K) - J_s(K^\star) \le \lambda$, or equivalently
\begin{equation*}
\mathrm{tr}((P-P^\star)V)\le \lambda,
\end{equation*}
where $K^\star=-R^{-1}B^\mathrm{T}P^\star$, and $P^\star\succ0$ is the positive definite solution of the algebraic Riccati equation
\begin{equation}\label{eq:realARE}
A^\mathrm{T}P^\star+P^\star A+Q-P^\star BR^{-1}B^\mathrm{T}P^\star=0.
\end{equation}
We summarize the overall problem as follows.
\begin{problem}\label{pr:tr}
Solve the optimization problem:
\begin{equation}\nonumber
\begin{split}
\min_{K,P,W} \quad &\mathrm{tr}(WV)\\
\textrm{s.t.}\quad &(A+BK)^\mathrm{T}P+P(A+BK)+Q+K^\mathrm{T}RK=0,\\
&(A+BK)^\mathrm{T}W+W(A+BK)+C^\mathrm{T}C=0,\\
&\mathrm{tr}((P-P^\star)V)\le \lambda,~P\succeq0,~W\succeq0.
\end{split}
\end{equation}
\end{problem}

The second problem of co-optimizing \eqref{eq:LQint} and the second observability metric in \eqref{eq:obsfuns} is more intricate. The main reason is that when the pair $(A+BK,C)$ is not observable, the quantity $\mathrm{tr}(W^{-1}V^{-1})$ becomes undefined due to the observability Gramian $W$ being singular. This technicality commonly appears in controllability and observability quantification problems, and is bypassed by regularizing the observability Gramian with a small positive definite matrix \cite{guo2020actuator}. Hereon, we adopt this standard regularization approach, but instead of applying the regularization directly to $W$ after its computation, we apply it at the level of the  Lyapunov equation \eqref{eq:lyapW},  from which $W$ is obtained. As we will see later, this choice is crucial for obtaining convergent algorithmic solutions. We thus obtain the following alternative optimization for designing $K$.

\begin{problem}\label{pr:trinv}
Select $\epsilon>0$ and solve the optimization problem:
\begin{equation}\nonumber
\begin{split}
\min_{K,P,W_\epsilon} \quad &-\mathrm{tr}(W_\epsilon^{-1}V^{-1})\\
\textrm{s.t.}\quad &(A+BK)^\mathrm{T}P+P(A+BK)+Q+K^\mathrm{T}RK=0,\\
&(A+BK)^\mathrm{T}W_\epsilon+W_\epsilon(A+BK)+C^\mathrm{T}C+\epsilon I=0,\\
& \mathrm{tr}((P-P^\star)V)\le \lambda, ~P\succeq0,~W_\epsilon\succeq0.
\end{split}
\end{equation}
\end{problem}

As discussed in \cite{guo2020actuator}, the regularization parameter
$\epsilon>0$ makes the inverse-Gramian objective well defined. The
following proposition quantifies the approximation error introduced by
this regularization and shows that the regularized objective converges
to its unregularized counterpart as $\epsilon\rightarrow0$.
\begin{proposition}\label{prop:epsilon}
Let $
c_\epsilon(K):=-\mathrm{tr}(W_\epsilon^{-1}V^{-1}),$
where $W_\epsilon$ is the unique solution of $
(A+BK)^\mathrm{T}W_\epsilon+W_\epsilon(A+BK)+C^\mathrm{T}C+\epsilon I=0$ when $A+BK$ is Hurwitz. Define
\begin{equation}
c_0(K):=
\begin{cases}
-\mathrm{tr}(W_0^{-1}V^{-1}), & W_0\succ0,\\
-\infty, & W_0 \text{ is singular}.
\end{cases}
\end{equation}
\begin{enumerate}
\item If $W_0\succ0$, then $
c_\epsilon(K)=c_0(K)+O(\epsilon).$
\item If $W_0$ is singular, then $\lim_{\epsilon\rightarrow0}c_\epsilon(K)=c_0(K).$
\end{enumerate}
\end{proposition}
\begin{proof}
Subtracting the equations defining $W_\epsilon$ and $W_0$ yields $
W_\epsilon=W_0+\epsilon S,$
where $S\succ0$ is the unique solution of $
(A+BK)^\mathrm{T}S+S(A+BK)+I=0.$
Suppose first that $W_0\succ0$. Substituting the resolvent identity
$W_\epsilon^{-1}=W_0^{-1}-\epsilon\,W_0^{-1}SW_\epsilon^{-1}$ into itself yields
\begin{equation}
W_0^{-1}-W_\epsilon^{-1}
=\epsilon\,W_0^{-1}SW_0^{-1}
-\epsilon^2\,W_0^{-1}SW_\epsilon^{-1}SW_0^{-1}.
\end{equation}
Since $W_\epsilon\succeq W_0$ implies $W_\epsilon^{-1}\preceq W_0^{-1}$, and the last term is a congruence of $W_\epsilon^{-1}$ and hence positive semidefinite,
\[
0\preceq W_0^{-1}-W_\epsilon^{-1}
\preceq\epsilon\,W_0^{-1}SW_0^{-1}.
\]
Since $V^{-1}\succ0$, taking the trace against $V^{-1}$ gives
\begin{equation}
0\le c_\epsilon(K)-c_0(K)
\le
\epsilon\,
\mathrm{tr}(W_0^{-1}SW_0^{-1}V^{-1}),
\end{equation}
which proves the first claim.
Now suppose that $W_0$ is singular, so that $c_0(K)=-\infty$. Since
$S\succ0$, the minimum eigenvalue of $W_\epsilon=W_0+\epsilon S$
converges to zero as $\epsilon\rightarrow0$. Consequently,
$\mathrm{tr}(W_\epsilon^{-1})\rightarrow+\infty$, and because
$V^{-1}\succ0$, it follows that
$\mathrm{tr}(W_\epsilon^{-1}V^{-1})\rightarrow+\infty$. Hence,
$c_\epsilon(K)\rightarrow-\infty=c_0(K)$. \frQED
\end{proof}
\begin{remark}
Proposition~\ref{prop:epsilon} shows that the regularization error is
$O(\epsilon)$ when $W_0\succ0$, so decreasing $\epsilon$ reduces this
error linearly but increases the condition number of $W_\epsilon$.
Hence, $\epsilon$ should be chosen small enough for the regularization
error to be negligible, but large enough to maintain good numerical
conditioning.
\end{remark}

Both Problem \ref{pr:tr} and \ref{pr:trinv} involve bilinear constraints. Problem \ref{pr:trinv} is further complicated by its nonlinear cost function.
We will provide lower bounds to their values that showcase the fundamental limits of the trade-off between adversarial observability and performance, as it relates to $\lambda$. Moreover, we will derive an algorithm that exactly solves Problem \ref{pr:tr} using an equivalent SDP, and an algorithm that approximately solves Problem \ref{pr:trinv} through sequential SDP.

\section{Adversarial Observability and Performance Trade-offs}\label{sec:tradeoff}

In this section, we study the adversarial observability and performance trade-off in Problems \ref{pr:tr} and \ref{pr:trinv}. We specifically provide lower bounds to the values of these problems as a function of the trade-off parameter $\lambda$, and study their robustness to uncertainty in the adversary's sensing matrix.

\subsection{Quantitative Lower Bounds}

Define the matrices $Z^\star, ~S^\star, ~U^\star\succ0$ that solve
\begin{align}\label{eq:ZLE}
&(A+BK^\star)Z^\star+Z^\star(A+BK^\star)^\tr +V=0,\\
&(A+BK^\star)S^\star+S^\star(A+BK^\star)^\tr +I=0,\label{eq:SLE}\\
&(A+BK^\star)^\tr U^\star+ U^\star(A+BK^\star) +I=0.
\end{align}
We present an intermediate result that bounds the norm of the gain distance $\tilde{K}:=K-K^\star$, where $K$ is the solution to either Problem \ref{pr:tr} or \ref{pr:trinv}, and $K^\star$ is the optimal gain.

\begin{lemma}\label{le:aux}
It holds that $\|\tilde{K}\|\le f(\lambda)$, where
\begin{equation*}
f(\lambda):=\frac{\lambda\|Z^\star\|\|{V^{-1}}\| \norm{B}   }{\underline{\lambda}(Z^\star)\underline{\lambda}(R)}\left(1+\sqrt{1+\frac{\underline{\lambda}(Z^\star)\underline{\lambda}(R)}{\lambda\|Z^\star\|^2\|{V^{-1}}\|^2 \norm{B}^2}} \right).
\end{equation*}
\end{lemma}
\begin{proof}
 From first constraint of Problem \ref{pr:tr} or \ref{pr:trinv} and from \eqref{eq:realARE}:
\begin{align*}
&(A+BK)^\mathrm{T}P+P(A+BK)+Q+K^\mathrm{T}RK=0,\\
&(A+BK)^\mathrm{T}P^\star+P^\star(A+BK)\\&\qquad\qquad\qquad+Q+K^{\star^\mathrm{T}}RK^\star-\tilde{K}^\tr B^\tr P^\star-P^\star B\tilde{K}=0.
\end{align*}
Subtracting these two equations, using the property $B^\tr P^\star=-RK^\star$, and defining $\tilde{P}=P-P^\star$, we get
\begin{equation}\nonumber
(A+BK)^\mathrm{T}\tilde{P}+\tilde{P}(A+BK)+\tilde{K}^\tr R\tilde{K}=0,
\end{equation}
and, by adding and subtracting identical terms:
\begin{equation}\nonumber
(A+BK^\star)^\mathrm{T}\tilde{P}+\tilde{P}(A+BK^\star)+\tilde{K}^\tr R\tilde{K}+\tilde{K}^\tr B^\tr\tilde{P}+\tilde{P}B\tilde{K}=0.
\end{equation}
The implicit solution to this equation is given by:
\begin{equation*}
\tilde{P}=\int_0^\infty e^{(A+BK^\star)^\mathrm{T}t}(\tilde{K}^\tr R\tilde{K}+\tilde{K}^\tr B^\tr\tilde{P}+\tilde{P}B\tilde{K})e^{(A+BK^\star)t}\textrm{d}t.
\end{equation*} 
Using the constraint $\mathrm{tr}((P-P^\star)V)=\mathrm{tr}(\tilde{P}V)\le \lambda$, this implies
\begin{align}\nonumber
\mathrm{tr}(\tilde{P}V) &= \mathrm{tr}\Big({\int_0^\infty} e^{(A+BK^\star)^\mathrm{T}t}(\tilde{K}^\tr R\tilde{K}\\\nonumber&\qquad\qquad\qquad+\tilde{K}^\tr B^\tr\tilde{P}+\tilde{P}B\tilde{K})e^{(A+BK^\star)t}\textrm{d}tV\Big)\\&\nonumber= \mathrm{tr}\Big({\int_0^\infty} e^{(A+BK^\star)t}Ve^{(A+BK^\star)^\mathrm{T}t}\textrm{d}t (\tilde{K}^\tr R\tilde{K}\\&\nonumber\qquad\qquad\qquad+\tilde{K}^\tr B^\tr\tilde{P}+\tilde{P}B\tilde{K})\Big)\\&\nonumber=\mathrm{tr}(Z^\star(\tilde{K}^\tr R\tilde{K}+\tilde{K}^\tr B^\tr\tilde{P}+\tilde{P}B\tilde{K}))\\&=\mathrm{tr}(Z^\star\tilde{K}^\tr R\tilde{K})+2\mathrm{tr}(\tilde{P}B\tilde{K}Z^\star)\le \lambda,\label{eq:th1c}
\end{align}
where $Z^\star$ solves the Lyapunov equation \eqref{eq:ZLE}.
Note that since $P^\star$ is the ARE solution \eqref{eq:realARE}, it follows that $P^\star\preceq P$, hence  $\tilde{P}\succeq0$. Therefore, using Fact 8.12.28 from \cite{bernstein2009matrix}, \eqref{eq:th1c} yields
\begin{equation}\label{eq:th1d}
\mathrm{tr}(Z^\star\tilde{K}^\tr R\tilde{K})-2\mathrm{tr}(\tilde{P})\|B\tilde{K}Z^\star \|\le \lambda.
\end{equation}
Moreover, we have 
\begin{equation*}
\mathrm{tr}(\tilde{P})=\mathrm{tr}(V^{1/2}\tilde{P}V^{1/2}V^{-1})\le \mathrm{tr}(\tilde{P}V)\|V^{-1}\|\le \lambda \|V^{-1}\|. 
\end{equation*}
Hence, \eqref{eq:th1d} yields
\begin{equation*}
\underline{\lambda}(Z^\star)\underline{\lambda}(R)\|\tilde{K}\|^2-2\lambda\|Z^\star\|\|{V^{-1}}\| \norm{B}\|\tilde{K}\|-\lambda\le0.
\end{equation*}
Solving this quadratic inequality we obtain
\begin{align*}
\|\tilde{K}\|&\le \frac{\lambda\|Z^\star\|\|{V^{-1}}\| \norm{B}   }{\underline{\lambda}(Z^\star)\underline{\lambda}(R)}\left(1+\sqrt{1+\frac{\underline{\lambda}(Z^\star)\underline{\lambda}(R)}{\lambda\|Z^\star\|^2\|{V^{-1}}\|^2 \norm{B}^2}} \right),
\end{align*}
which is the required result. \frQED
\end{proof}

Next, denote the value of Problem \ref{pr:tr} with respect to $\lambda$ as $J_1(\lambda)$, and that of Problem \ref{pr:trinv} as $J_2(\lambda)$. Using Lemma \ref{le:aux}, we provide lower bounds to the values of these problems as functions of the observability-performance trade-off parameter $\lambda$.

\begin{theorem}\label{th:lb}
The following hold true.
\begin{enumerate}
\item For all $\lambda\ge0$:
\begin{equation}\label{eq:trade1}
J_1(\lambda)\ge \frac{1}{1+2\mathrm{tr}(Z^\star)f(\lambda)\|V^{-1}\|\norm{B}}J_1(0),
\end{equation}
\item There exists $\lambda^\star>0$ such that for all $\lambda\le \lambda^\star$ 
\begin{equation}\label{eq:trade2}
J_2(\lambda)\ge \frac{1}{1-\frac{2\kappa(V)f(\lambda)\norm{V}\norm{B}\|U^\star\|\mathrm{tr}(W_\epsilon^0)\mathrm{tr}((W^0_\epsilon)^{-1})}{1-2\mathrm{tr}(S^\star)f(\lambda)\|B\|}}J_2(0),
\end{equation}
where $W_\epsilon^0$ corresponds to $W_\epsilon$ under $K=K^\star$.
\end{enumerate}
\end{theorem}
\begin{proof}
To prove item 1, the second constraint of Problem \ref{pr:tr} yields
\begin{equation*}
(A+BK^\star)^\mathrm{T}W+W(A+BK^\star)+C^\mathrm{T}C+WB\tilde{K}+\tilde{K}^{\mathrm{T}}B^\tr W=0.
\end{equation*}
Therefore, denoting the Gramian in Problem \ref{pr:tr} resulting with $\lambda=0$ as $W^0$, we obtain
\begin{equation*}
W=W^0+\int_0^\infty e^{(A+BK^\star)^\mathrm{T}t}(WB\tilde{K}+\tilde{K}^{\mathrm{T}}B^\tr W)e^{(A+BK^\star)t}\mathrm{d}t.
\end{equation*}
Applying traces to each side of this equality: 
\begin{align}\nonumber
\mathrm{tr}&(WV)=\mathrm{tr}(W^0V)\\&\nonumber+\mathrm{tr}\left(\int_0^\infty e^{(A+BK^\star)t}Ve^{(A+BK^\star)^\mathrm{T}t}\textrm{d}t(WB\tilde{K}+\tilde{K}^{\mathrm{T}}B^\tr W) \right)\\&=\mathrm{tr}(W^0V)+\mathrm{tr}(Z^\star(WB\tilde{K}+\tilde{K}^{\mathrm{T}}B^\tr W)).\label{eq:similar}
\end{align}
 Using Fact 8.12.28 from \cite{bernstein2009matrix} and Lemma \ref{le:aux}
 \begin{align}\nonumber
\mathrm{tr}(WV)&\ge\mathrm{tr}(W^0V)-2\mathrm{tr}(Z^\star)\norm{W}\norm{B}\|\tilde{K}\|\\\nonumber&\ge\mathrm{tr}(W^0V)-2\mathrm{tr}(Z^\star)f(\lambda)\mathrm{tr}(WV)\|V^{-1}\|\norm{B}.
 \end{align}
Rearranging terms, we get  $\mathrm{tr}(WV)\ge \mathrm{tr}(W^0V)/(1+2\mathrm{tr}(Z^\star)f(\lambda)\|V^{-1}\|\norm{B})$, which proves item 1.

 To prove item 2, because of the second constraint of Problem \ref{pr:trinv}, we have similarly that $W_\epsilon=W^0_\epsilon+X,$
where $W_{\epsilon}^0$ denotes the Gramian in Problem \ref{pr:trinv} for $\lambda=0$, and $X=\int_0^\infty e^{(A+BK^\star)^\mathrm{T}t}(W_\epsilon B\tilde{K}+\tilde{K}^{\mathrm{T}}B^\tr  W_\epsilon)e^{(A+BK^\star)t}\mathrm{d}t$. Using the Woodbury identity we have
$W_\epsilon^{-1}V^{-1}=(W^0_\epsilon)^{-1}V^{-1}-(W^0_\epsilon)^{-1}XW_\epsilon^{-1}V^{-1}$, hence
\begin{align}\nonumber
\mathrm{tr}(W_\epsilon^{-1}&V^{-1})=\mathrm{tr}((W^0_\epsilon)^{-1}V^{-1})-\mathrm{tr}((W^0_\epsilon)^{-1}XW_\epsilon^{-1}V^{-1})
\\&\nonumber\le \mathrm{tr}((W^0_\epsilon)^{-1}V^{-1})+\kappa(V)\mathrm{tr}((W^0_\epsilon)^{-1}V^{-1})\|XW_\epsilon^{-1}\|
\\& \le (1+\kappa(V)\|W_\epsilon^{-1}\|\norm{X})\mathrm{tr}((W^0_\epsilon)^{-1}V^{-1}).\label{eq:th1temp}
\end{align}
Moreover, by the definition of $X$, we have $X\preceq 2\|\tilde{K}\|\norm{B}\norm{W_\epsilon}U^\star\preceq 2f(\lambda)\norm{B}\norm{W_\epsilon}U^\star$. Combining this with \eqref{eq:th1temp} yields
\begin{align*}
&\mathrm{tr}(W_\epsilon^{-1}V^{-1})\\& \le (1{+}2\kappa(V)f(\lambda)\norm{B}\norm{W_\epsilon}\|W_\epsilon^{-1}\|\|U^\star\|)\mathrm{tr}((W^0_\epsilon)^{-1}V^{-1})
\\&\le (1+2\kappa(V)f(\lambda)\norm{B}\mathrm{tr}(W_\epsilon)\mathrm{tr}(W_\epsilon^{-1})\|U^\star\|)\mathrm{tr}((W^0_\epsilon)^{-1}V^{-1}).
\end{align*}
Regrouping terms:
\begin{multline}
\mathrm{tr}(W_\epsilon^{-1}V^{-1})\\\le \frac{\mathrm{tr}((W_\epsilon^0)^{-1}V^{-1})}{1-2\kappa(V)f(\lambda)\norm{V}\norm{B}\|U^\star\|\mathrm{tr}(W_\epsilon)\mathrm{tr}((W^0_\epsilon)^{-1}V^{-1})}.\label{eq:th1end1}
\end{multline}
Finally, similar to \eqref{eq:similar} in item 1 we have
\begin{align*}
\mathrm{tr}(W_\epsilon )&=\mathrm{tr}(W^0_\epsilon)+\mathrm{tr}(S^\star(W_\epsilon B\tilde{K}+\tilde{K}^{\mathrm{T}}B^\tr W_\epsilon)))
\\&\le \mathrm{tr}(W^0_\epsilon)+2\mathrm{tr}(S^\star)\|W_\epsilon B\tilde{K}\|
\\&\le \mathrm{tr}(W^0_\epsilon)+2\mathrm{tr}(S^\star)f(\lambda)\|B\|\mathrm{tr}(W_\epsilon),
\end{align*}
where we used $\|\tilde{K}\|\le f(\lambda)$. Therefore
\begin{equation}\label{eq:th1end2}
\mathrm{tr}(W_\epsilon)\le \frac{1}{1-2\mathrm{tr}(S^\star)f(\lambda)\|B\|}\mathrm{tr}(W_\epsilon^0).
\end{equation}
Combining equations \eqref{eq:th1end1}-\eqref{eq:th1end2} yields
\begin{equation}\nonumber
\mathrm{tr}(W_\epsilon^{-1}V^{-1})\le \frac{\mathrm{tr}((W_\epsilon^0)^{-1}V^{-1})}{1{-}\frac{2\kappa(V)f(\lambda)\norm{V}\norm{B}\|U^\star\|\mathrm{tr}(W_\epsilon^0)\mathrm{tr}((W^0_\epsilon)^{-1}V^{-1})}{1-2\mathrm{tr}(S^\star)f(\lambda)\|B\|}},
\end{equation}
which is equivalent to the final result. \frQED
\end{proof}

Note that $f(\lambda)\sim O(\lambda)$ for large $\lambda$ and $f(\lambda)\sim O(\sqrt{\lambda})$ for small $\lambda$. Looking at \eqref{eq:trade1}, this implies that small performance losses $\lambda$ yield steeper trade-offs in reducing adversarial observability for the cost \eqref{eq:trade1}. On the other hand, the lower bound in  Problem \ref{pr:trinv} can become unbounded as $\lambda$ increases, since a larger $\lambda$ allows the Gramian inverse to become singular. As a result, higher values of $\lambda$ can produce sharper improvements in the objective of Problem~\ref{pr:trinv}.  An illustration of the function $f(\lambda)$ is provided in Figure \ref{fig:fl} for the first example of the simulations in Section \ref{sec:sim}. 

The lower bounds also depend on a number of other parameters. Most notable is $\norm{B}$, which dictates that larger actuation authority for the operator allows for sharper decrease in adversarial observability.  They also depend on $V$, which captures the interplay between observability minimization and initial condition covariance.

Next, note that although the lower bound of Problem \ref{pr:trinv} is local, the objective value of Problem \ref{pr:trinv} remains globally bounded due to the regularization parameter $\epsilon$. Consequently, we can extend the lower bound to hold globally as follows.
\begin{proposition}\label{prop:loose}
    For all $\lambda\ge0$:
    \begin{equation*}
        J_2(\lambda)\ge -\frac{2\mathrm{tr}(V^{-1})(\norm{A+BK^\star}+\norm{B}f(\lambda))}{\epsilon}
    \end{equation*}
\end{proposition}
\begin{proof}
From the second constraint of Problem \ref{pr:trinv} we obtain
\begin{align*}
W_\epsilon&=\int_0^\infty e^{(A+BK)^\mathrm{T}t}(C^\mathrm{T}C+\epsilon I)e^{(A+BK)t}\textrm{d}t\\&\ge \epsilon\int_0^\infty e^{(A+BK)^\mathrm{T}t}e^{(A+BK)t}\textrm{d}t\\&\ge\epsilon\int_0^\infty e^{-2\norm{A+BK}t}\textrm{d}t=\frac{\epsilon}{2\norm{A+BK}},
\end{align*}
hence $\underline{\lambda}(W_\epsilon)\ge \frac{\epsilon}{2\norm{A+BK}}$. Therefore:
\begin{multline}\label{eq:loose1}
\mathrm{tr}(W_\epsilon^{-1}V^{-1})\le \mathrm{tr}(V^{-1})\bar{\lambda}(W_\epsilon^{-1})\\=\frac{\mathrm{tr}(V^{-1})}{\underline{\lambda}(W_\epsilon)}\le \frac{2\mathrm{tr}(V^{-1})\norm{A+BK}}{\epsilon}.
\end{multline}
Moreover, from Lemma \ref{le:aux}:
\begin{equation}\label{eq:loose2}
\norm{A{+}BK}\le\|A{+}BK^\star{+}B\tilde{K}\|\le\norm{A{+}BK^\star}+\norm{B}f(\lambda).
\end{equation}
Combining \eqref{eq:loose1}-\eqref{eq:loose2} yields the required result. \frQED
\end{proof}
Combining Theorem \ref{th:lb} with Proposition \ref{prop:loose} yields lower bounds for Problems \ref{pr:tr} and \ref{pr:trinv} across all values of the trade-off parameter $\lambda$ that quantify the fundamental limits of adversarial observability reduction under a given level of performance degradation.

\begin{figure}[!t] 
		\centering
                \includegraphics[width=6cm,height=3cm]{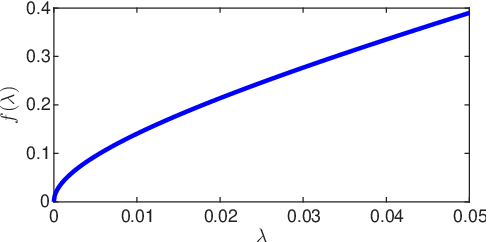}
                \caption{\small Illustration of the function $f(\lambda)$. }
                \label{fig:fl}
\end{figure}

\subsection{Robustness to Adversary's Sensing Uncertainty}

Assumption \ref{ass:assym} requires knowledge of the adversary’s sensing matrix 
$C$. This is often realistic in practice, for instance, when the adversary employs radar with known location and measurement geometry to track an aircraft  or spacecraft. In some settings, however, the system operator may only know a superset of the adversary’s possible sensing configurations, rather than the exact matrix $C$.

Solving Problems \ref{pr:tr} and \ref{pr:trinv} using such a superset generally yields conservative (and thus suboptimal) solutions. Nevertheless, any controller obtained in this manner provides a certified upper bound on the adversary’s observability metric, as we establish next.

\begin{proposition}\label{prop:robustness}
    Assume that the adversary's sensing matrix is $\hat{C}\in\mathbb{R}^{\hat{m}\times n}$, where
    $\hat{m}\le m$, and where the rows of $\hat{C}$ are a subset of the rows of $C$. Then, $J_{oj}(K;~\hat{C})\le J_{oj}(K;~C)$, $j=1,2$.
\end{proposition}

\begin{proof}
Denote as $W_C$ and $W_{\hat{C}}$ the observability Gramians corresponding to $C$ and $\hat{C}$. Then, by definition, we have $(A+BK)^\mathrm{T}W_C+W_C(A+BK)+C^\mathrm{T}C=0$ and $(A+BK)^\mathrm{T}W_{\hat{C}}+W_{\hat{C}}(A+BK)+\hat{C}^\mathrm{T}\hat{C}=0$. Since $\hat{C}^\textrm{T}\hat{C}=C^\textrm{T}S^\textrm{T}SC$ for a row selection matrix $S\in\mathbb{R}^{\hat{m}\times m}$, we have $C^\textrm{T}C-\hat{C}^\textrm{T}\hat{C}=C^\textrm{T}(I-S^\textrm{T}S)C\succeq 0$. Hence, $W_{\hat{C}}\preceq W_{C}$, and thus $J_{o1}(K;~\hat{C})=\textrm{tr}(W_{\hat{C}} V)  \le\textrm{tr}(W_C V)= J_{o1}(K;~C)$. The fact that $J_{o2}(K;~\hat{C})\le J_{o2}(K;~C)$ follows similarly since negated inversion preserves matrix ordering. \frQED 
\end{proof}

Another possible sensing uncertainty concerns perturbations in the numerical values of $C$, rather than uncertainty in its rows (e.g., $\hat{C}=C+\delta C$ for some small $\delta C$). In such cases, Problem \ref{pr:tr} is inherently more robust since its objective depends quadratically on $C$ (see Corollary \ref{cor:Csens}), whereas Problem~\ref{pr:trinv} can be more sensitive to perturbations in $C$, particularly when the Gramian approaches singularity due to the inverse-based objective. One could also address such uncertainty in the worst-case sense through a robust SDP formulation, though this remains beyond the scope of this work.

\section{Algorithmic Solutions to Problems \ref{pr:tr} and \ref{pr:trinv}}\label{sec:sdp}

We subsequently provide algorithmic solutions to 
Problems \ref{pr:tr} and \ref{pr:trinv}. Specifically, we provide an exact algorithmic solution for Problem \ref{pr:tr} by reformulating it into an equivalent SDP, and an approximate algorithmic solution for Problem \ref{pr:trinv} using sequential SDP.

\subsection{Solution to Problem \ref{pr:tr} using an SDP}

We begin by showing that Problem \ref{pr:tr} can be cast as an SDP.
\begin{theorem}\label{th:SDP}
The solution to Problem \ref{pr:tr} is given by $K=XS^{-1}$, where $X$ and $S$ solve the SDP
\begin{equation}\label{eq:prog4}
\begin{split}
\min_{X,S,Z} \quad &\mathrm{tr}(CSC^\mathrm{T})\\
\mathrm{s.t.}\quad &AS+BX+SA^\mathrm{T}+X^\mathrm{T}B^\mathrm{T}+V= 0,\\
&\mathrm{tr}(SQ)+\mathrm{tr}(Z)\le \mathrm{tr}(P^\star V)+\lambda,\\
&\begin{bmatrix}Z & R^{1/2}X \\ X^\mathrm{T}R^{1/2} & S \end{bmatrix}\succeq0, ~Z\succeq0,~S\succeq0.
\end{split}
\end{equation}
\end{theorem}
\begin{proof}
By the constraints of Problem \ref{pr:tr}, we have $W=\int_0^{\infty}e^{(A+BK)^\textrm{T}\tau}C^\mathrm{T}Ce^{(A+BK)\tau}\textrm{d}\tau$ and $P=\int_0^{\infty}e^{(A+BK)^\mathrm{T}\tau}(Q+K^\mathrm{T}RK)e^{(A+BK)\tau}\textrm{d}\tau$. Therefore, using the cyclic property of the trace operator, we obtain $\mathrm{tr}(WV)=\mathrm{tr}(CSC^\mathrm{T})$ and $\mathrm{tr}(PV)=\mathrm{tr}(S(Q+K^\mathrm{T}QK))$, where $S\succ0$ is the unique solution of the Lyapunov equation $(A+BK)S+S(A+BK)^\mathrm{T}+V=0.$
Using these equivalences, we can reduce the dimension of Problem \ref{pr:tr} and conclude it is equivalent to the program
\begin{equation}\label{eq:prog1}
\begin{split}
\min_{K,S} \quad &\mathrm{tr}(CSC^\mathrm{T})\\
\textrm{s.t.}\quad &(A+BK)S+S(A+BK)^\mathrm{T}+V=0,\\
&\mathrm{tr}(S(Q+K^\mathrm{T}RK))\le \mathrm{tr}(P^\star V)+\lambda,~S\succeq0.
\end{split}
\end{equation}
Performing the change of variables $X=KS \Longrightarrow K=XS^{-1}$, it follows that \eqref{eq:prog1} has the same optimal solution as
\begin{equation}\label{eq:prog2}
\begin{split}
\min_{X,S} \quad &\mathrm{tr}(CSC^\mathrm{T})\\
\textrm{s.t.}\quad &AS+BX+SA^\mathrm{T}+X^\mathrm{T}B^\mathrm{T}+V=0,~ S\succeq0,\\
&\mathrm{tr}(SQ)+\mathrm{tr}(R^{1/2}XS^{-1}X^\mathrm{T}R^{1/2})\le \mathrm{tr}(P^\star V)+\lambda.
\end{split}
\end{equation}
Introducing the variable $Z\succeq R^{1/2}XS^{-1}X^\mathrm{T}R^{1/2}$ preserves the optimal solution. Therefore, \eqref{eq:prog2} is equivalent to
\begin{equation}\label{eq:prog3}
\begin{split}
\min_{X,S,Z} \quad &\mathrm{tr}(CSC^\mathrm{T})\\
\textrm{s.t.}\quad &AS+BX+SA^\mathrm{T}+X^\mathrm{T}B^\mathrm{T}+V= 0,\\
&\mathrm{tr}(SQ)+\mathrm{tr}(Z)\le \mathrm{tr}(P^\star V)+\lambda,\\
&Z- R^{1/2}XS^{-1}X^\mathrm{T}R^{1/2}\succeq0, Z\succeq0, S\succeq0.
\end{split}
\end{equation}
Finally, using the Schur complement, we obtain the SDP \eqref{eq:prog4}. \frQED
\end{proof}

Theorem \ref{th:SDP} recasts Problem \ref{pr:tr} as an SDP, which can be
solved using standard solvers as summarized in
Algorithm \ref{al:SDP}. Moreover, the SDP reformulation isolates the
dependence on the sensing matrix $C$ in the objective function, hence
enabling a characterization of the sensitivity of its optimal
value to perturbations in $C$.
\begin{corollary}\label{cor:Csens}
Let $J_1(\lambda;C)$ denote the optimal value of Problem~\ref{pr:tr} for sensing
matrix $C$. Then, for any perturbation $\delta C$,
\begin{equation}\nonumber
|J_1(\lambda;C+\delta C)-J_1(\lambda;C)|\le 2M_\lambda\|C\|\,\|\delta C\|+M_\lambda\|\delta C\|^2,
\end{equation}
where $M_\lambda:=(\mathrm{tr}(P^\star V)+\lambda)/\underline{\lambda}(Q)$.
\end{corollary}
\begin{proof}
In \eqref{eq:prog4}, $C$ appears only in the objective
$\mathrm{tr}(CSC^\mathrm{T})$. Hence, 
\begin{equation}\nonumber
\mathrm{tr}((C{+}\delta C)S(C{+}\delta C)^\mathrm{T}
{-}CSC^\mathrm{T})
=
2\mathrm{tr}(\delta C\,SC^\mathrm{T})
+
\mathrm{tr}(\delta C\,S\delta C^\mathrm{T})
\end{equation}
for any feasible $S$. Since $S,Z\succeq0$ and $Q\succ0$, $\mathrm{tr}(Z)\ge0$ implies $\mathrm{tr}(SQ)\le \mathrm{tr}(P^\star V)+\lambda$, while $\mathrm{tr}(SQ)\ge\underline{\lambda}(Q)\mathrm{tr}(S)$. Hence, $\|S\|\le\mathrm{tr}(S)\le (\mathrm{tr}(P^\star V)+\lambda)/\underline{\lambda}(Q)=M_\lambda$.
 Therefore,
\begin{equation}\nonumber
|
\mathrm{tr}((C{+}\delta C)S(C{+}\delta C)^\mathrm{T}
{-}CSC^\mathrm{T})
|
\le
2M_\lambda\|C\|\,\|\delta C\|
+
M_\lambda\|\delta C\|^2.
\end{equation}
Since the right-hand side $\beta(\delta C):=2M_\lambda\|C\|\|\delta C\|+M_\lambda\|\delta C\|^2$ is independent of $S$, minimizing over feasible $S$ gives $J_1(\lambda;C+\delta C)\le J_1(\lambda;C)+\beta(\delta C)$; swapping the roles of $C$ and $C+\delta C$ gives the reverse inequality, hence $|J_1(\lambda;C+\delta C)-J_1(\lambda;C)|\le\beta(\delta C)$. \frQED
\end{proof}

\subsection{Approximate Solution to Problem \ref{pr:trinv} using Sequential SDP}

\begin{algorithm}[!t]
\caption{SDP for Problem \ref{pr:tr}}
\hspace*{1.7em} \textbf{Output}: \parbox[t]{\dimexpr\linewidth-1.5em}{Exact solution $(K,P,W)$ to Problem \ref{pr:tr}.}
\begin{algorithmic}[1]
\Procedure{}{}
\State Solve for $(X, S, Z)$ from \eqref{eq:prog4}.
\State Compute $K=XS^{-1}$.
\State Compute $P,~W$ from the equality constraints of Problem \ref{pr:tr}.
\EndProcedure
\end{algorithmic}\label{al:SDP}
\end{algorithm}

Since the negated trace of a matrix inverse is concave in the cone of positive definite matrices, Problem \ref{pr:trinv} is not convex and thus cannot be cast as an SDP. However, we can formulate an algorithm to approximately solve it using the convex-concave procedure.

Toward this end, one challenge entailed in Problem \ref{pr:trinv} is that its cost function is nonlinear and, most importantly, it decreases rapidly as the matrix $W_\epsilon$ approaches singularity. As a result, linearizing this cost can lead to large approximation errors that slow down the convergence of a convex-concave procedure.
Nevertheless, by transforming the Lyapunov equation for the observability Gramian into a special form of an algebraic Riccati equation, we can effectively deal with this issue. Specifically, we show Problem \ref{pr:trinv} can be cast as an equivalent  difference of convex functions (DC) program with a linear objective function that involves no inverses of matrix variables.

\begin{lemma}\label{le:relax_pr2}
Under the transformation $Y_{\epsilon}=W_{\epsilon}^{-1}$, Problem \ref{pr:trinv} admits the same optimal solution as the DC program with linear cost:
\begin{align}\label{eq:problem2b}
    \min_{K,P,Y_{\epsilon}}&  \quad-\mathrm{tr}(Y_{\epsilon}V^{-1})\\\nonumber
    \mathrm{s.t.}\quad &\begin{bmatrix}\frac{1}{2}(A+BK+P)^\mathrm{T}(A+BK+P)+Q & K^\mathrm{T} \\ \star & -R^{-1} \end{bmatrix}\\&-\begin{bmatrix}\frac{1}{2}(A+BK-P)^\mathrm{T}(A+BK-P) & 0 \\ \star & 0 \end{bmatrix}\preceq0,\nonumber\\\nonumber
    &\begin{bmatrix}\frac{1}{2}(A+BK+Y_{\epsilon})(A+BK+Y_{\epsilon})^\mathrm{T} & Y_{\epsilon}\sqrt{C^\mathrm{T}C{+}\epsilon I} \\ \star & I \end{bmatrix}\\&-\begin{bmatrix}\frac{1}{2}(A+BK-Y_{\epsilon})(A+BK-Y_{\epsilon})^\mathrm{T} & 0 \\ \star & 0 \end{bmatrix}\preceq0,\nonumber\\
    & \mathrm{tr}((P-P^\star)V)\le \lambda,~P\succeq0,~Y_{\epsilon}\succeq0.\nonumber
    \end{align}
\end{lemma}
\begin{proof}
We first rewrite the cost function of Problem \ref{pr:trinv} into a linear one. To that end,
note that since $P\succeq0$ and $Q+K^\mathrm{T}RK\succ0$, the first constraint of Problem \ref{pr:trinv} is a Lyapunov equation for $A+BK$, hence $A+BK$ is Hurwitz \cite{bernstein2009matrix}. Therefore,  the second constraint in Problem \ref{pr:trinv} implies $W_\epsilon\succ0$ strictly since $C^\textrm{T}C+\epsilon I\succ0$, hence $W_\epsilon^{-1}$ exists. In light of this, multiplying the second constraint of Problem \ref{pr:trinv} from the left and from the right with $W_\epsilon^{-1}$, we obtain the equation
\begin{equation}\label{eq:ARE}
-Y_{\epsilon}(A+BK)^\mathrm{T}-(A+BK)Y_{\epsilon}-Y_{\epsilon}(C^\mathrm{T}C+\epsilon I)Y_{\epsilon}=0,
\end{equation}
for which $Y_{\epsilon}=W_\epsilon^{-1}$ is a solution. However, \eqref{eq:ARE} also admits other positive semidefinite solutions. We thus need to show that using \eqref{eq:ARE} in place of the second constraint in Problem \ref{pr:trinv} preserves optimality. To this end, note that we can rewrite \eqref{eq:ARE}  as
\begin{multline*}
\hspace{-3mm}-Y_{\epsilon}(A+BK+Y_{\epsilon}(C^\mathrm{T}C+\epsilon I))^\mathrm{T}-(A+BK+Y_{\epsilon}(C^\mathrm{T}C+\epsilon I))Y_{\epsilon}\\+Y_{\epsilon}(C^\mathrm{T}C+\epsilon I)Y_{\epsilon}=0.
\end{multline*}
Since $C^\mathrm{T}C+\epsilon I\succ0$, for $Y_{\epsilon}=W_\epsilon^{-1}\succ0$ we have $Y_{\epsilon}(C^\mathrm{T}C+\epsilon I)Y_{\epsilon}\succ0$. By the Lyapunov Theorem, $-(A+BK+Y_{\epsilon}(C^\mathrm{T}C+\epsilon I))^\mathrm{T}$ is strictly stable for $Y_{\epsilon}=W_\epsilon^{-1}$, hence $Y_{\epsilon}=W_\epsilon^{-1}$ is the stabilizing solution to \eqref{eq:ARE}, and hence also the maximal one by Theorem 12.18.4 in \cite{bernstein2009matrix}. Therefore, for any other solution $\hat{Y}_{\epsilon}\ne W_{\epsilon}^{-1}$ to \eqref{eq:ARE} we have $\hat{Y}_{\epsilon}\prec W_{\epsilon}^{-1}$ and thus $-\mathrm{tr}(\hat{Y}_{\epsilon}V^{-1})>-\mathrm{tr}(W_{\epsilon}^{-1}V^{-1})$. It thus follows that if we substitute the second constraint in Problem \ref{pr:trinv} with \eqref{eq:ARE} and perform the variable change $Y_{\epsilon}=W_{\epsilon}^{-1}$, the optimal solution will remain the same. Hence, Problem \ref{pr:trinv} is equivalent to the following optimization problem with linear cost:
\begin{align}\label{eq:pr2temp}
\min_{K,P,Y_{\epsilon}}& \quad -\mathrm{tr}(Y_{\epsilon}V^{-1})\\
\textrm{s.t.}\quad &(A+BK)^\mathrm{T}P+P(A+BK)+Q+K^\mathrm{T}RK=0,\nonumber\\
&-Y_{\epsilon}(A+BK)^\mathrm{T}-(A+BK)Y_{\epsilon}-Y_{\epsilon}(C^\mathrm{T}C+\epsilon I)Y_{\epsilon}=0,\nonumber\\
& \mathrm{tr}((P-P^\star)V)\le \lambda, ~ P\succeq0,~Y_{\epsilon}\succeq0.\nonumber
\end{align}

Subsequently, it is straightforward that relaxing the first equality constraint in \eqref{eq:pr2temp} into a less-than-equal-to inequality preserves the optimal solution and the stability of $A+BK$. In addition, since $\mathrm{tr}(Y_{\epsilon}V^{-1})$ appears negatively in the cost function, we can also relax the second constraint in \eqref{eq:pr2temp} into a greater-than-equal-to inequality. Employing these relaxations and using the Schur complement, it follows that \eqref{eq:pr2temp} is equivalent to
\begin{align}\nonumber
\min_{K,P,Y_{\epsilon}} &\quad -\mathrm{tr}(Y_{\epsilon}V^{-1})\\
\textrm{s.t.}\quad &\begin{bmatrix}(A+BK)^\mathrm{T}P+P(A+BK)+Q & K^\mathrm{T} \\ \star & -R^{-1} \end{bmatrix}\preceq0,\nonumber\\
&\begin{bmatrix}Y_{\epsilon}(A+BK)^\mathrm{T}+(A+BK)Y_{\epsilon} & Y_{\epsilon}\sqrt{C^\mathrm{T}C{+}\epsilon I} \\ \star & -I \end{bmatrix}\preceq0,\nonumber\\
& \mathrm{tr}((P-P^\star)V)\le \lambda,~P\succeq0,~Y_{\epsilon}\succeq0.\nonumber
\end{align}
The final result follows by using $Y_{\epsilon}(A+BK)^\mathrm{T}+(A+BK)Y_{\epsilon}=\frac{1}{2}(A+BK+Y_{\epsilon})(A+BK+Y_{\epsilon})^\mathrm{T}-\frac{1}{2}(A+BK-Y_{\epsilon})(A+BK-Y_{\epsilon})^\mathrm{T}$ and $(A+BK)^\mathrm{T}P+P(A+BK)=\frac{1}{2}(A+BK+P)^\mathrm{T}(A+BK+P)-\frac{1}{2}(A+BK-P)^\mathrm{T}(A+BK-P)$. \frQED
\end{proof}

Let us now define the functions\footnote{We use the subscript $j$ with a slight abuse of notation.}:
\begin{align*}
L_j(X,Z)&=\frac{1}{2}(A+BX_j-Z_j)^\mathrm{T}(A+BX_j-Z_j)\\&+\frac{1}{2}(A+BX_j-Z_j)^\mathrm{T}(B(X-X_j)-(Z-Z_j))\\&+\frac{1}{2}(B(X-X_j)-(Z-Z_j))^\mathrm{T}(A+BX_j-X_j),\\
L_j'(X,Z)&=\frac{1}{2}(A+BX_j-Z_j)(A+BX_j-Z_j)^\mathrm{T}\\&+\frac{1}{2}(A+BX_j-Z_j)(B(X-X_j)-(Z-Z_j))^\mathrm{T}\\&+\frac{1}{2}(B(X-X_j)-(Z-Z_j))(A+BX_j-X_j)^\mathrm{T}.
\end{align*}
To use the convex-concave procedure and obtain an approximate solution to Problem \ref{pr:trinv}, we linearize the concave parts of the constraints of \eqref{eq:problem2b} about a point $(K_j, P_j, W_j)$, employ the Schur complement, and obtain
  \begin{align}\label{eq:problem2semi}
    \min_{K,P,Y_{\epsilon}}  &\quad-\mathrm{tr}(Y_{\epsilon}V^{-1})\\\nonumber
    \mathrm{s.t.}\quad &\begin{bmatrix}Q-L_j(K,P) & K^\mathrm{T} & \frac{1}{\sqrt{2}}(A+BK+P)^\mathrm{T}  \\ \star & -R^{-1} & 0 \\ \star & \star & -I \end{bmatrix}\preceq0,\nonumber\\\nonumber
    &\hspace{-6mm}\begin{bmatrix}-L_j'(K,Y_{\epsilon}) & Y_{\epsilon}\sqrt{C^\mathrm{T}C+\epsilon I} & \frac{1}{\sqrt{2}}(A+BK+Y_{\epsilon})  \\ \star & -I & 0 \\ \star & \star & -I \end{bmatrix}\preceq0,\nonumber\\
    & \mathrm{tr}((P-P^\star)V)\le \lambda,~P\succeq0,~Y_{\epsilon}\succeq0,\nonumber
    \end{align}  
The iterative solution of this program, as shown in Algorithm \ref{al:SSPinv} and proved below, should eventually lead us to a stationary point of \eqref{eq:problem2b}.

\begin{algorithm}[!t]
\caption{Sequential SDP for Problem \ref{pr:trinv}}
\hspace*{1.7em}\textbf{Input}: \parbox[t]{\dimexpr\linewidth-4.5em}{Feasible point $K_0,P_0,Y_{\epsilon,0}$ to Problem \ref{pr:trinv}, tolerance $\delta>0$.}\\
\hspace*{\algorithmicindent} \textbf{Output}: \parbox[t]{\dimexpr\linewidth-1.5em}{Estimated solution $(\hat{K},\hat{P},\hat{Y}_{\epsilon})$ to Problem \ref{pr:trinv}.}
\begin{algorithmic}[1]
\Procedure{}{}
\State $j\leftarrow 0$.
\While{$j=0$ or $|\mathrm{tr}(Y_{\epsilon,j}-Y_{\epsilon,j-1})|\ge\delta$}
\State Solve for $K, P, Y_{\epsilon}$ from \eqref{eq:problem2semi}.
\State $(K_{j+1}, P_{j+1}, Y_{\epsilon,j+1})\leftarrow (K, P, Y_{\epsilon})$.
\State $j\leftarrow j+1$.
\EndWhile
\State $(\hat{K}, \hat{P}, \hat{Y}_{\epsilon})\leftarrow (K, P, Y_{\epsilon})$.
\EndProcedure
\end{algorithmic}\label{al:SSPinv}
\end{algorithm}

\begin{remark}
Algorithm~\ref{al:SSPinv} requires initialization with sets of matrices $K_0, P_0, Y_{\epsilon,0}$ that are feasible for Problem~\ref{pr:trinv}.  One such initialization is obtained by selecting $K_0$ as the linear-quadratic gain corresponding to $\lambda = 0$. The matrices $P_0$ and $Y_{\epsilon,0}$ are then obtained by solving the first and second Lyapunov equations of Problem~\ref{pr:trinv} under $K = K_0$, respectively, with $Y_{\epsilon,0}$ computed as the inverse of the resulting $W_\epsilon$. With such an initialization, \eqref{eq:problem2semi} will also be feasible. 
\end{remark}

\begin{theorem}\label{th:theorem}
Algorithm \ref{al:SSPinv} converges to a stationary point of \eqref{eq:problem2b}.
\end{theorem}

\begin{proof}
By Lemma \ref{le:relax_pr2}, Problem \ref{pr:trinv} is equivalent to \eqref{eq:problem2b}, a DC program whose functions are continuously differentiable in $K, P, Y_{\epsilon}$. By Lemma \ref{le:aux} and the constraint $\textrm{tr}((P-P^\star)V)\le\lambda$, the feasible set for $K$ and $P$ is bounded. The second constraint in \eqref{eq:problem2b} also implies
\begin{equation}\label{eq:aretemp}
-Y_{\epsilon}(A+BK)^\mathrm{T}-(A+BK)Y_{\epsilon}-Y_{\epsilon}(C^\mathrm{T}C+\epsilon I)Y_{\epsilon}\succeq0.
\end{equation}
Since $K$ is bounded and $C^\mathrm{T}C+\epsilon I\succ0$, the quadratic term
$Y_{\epsilon}(C^\mathrm{T}C+\epsilon I)Y_{\epsilon}$ dominates the linear terms
as $\|Y_{\epsilon}\|\rightarrow\infty$, so \eqref{eq:aretemp} cannot hold.
Hence, the feasible set for $Y_{\epsilon}$ is also bounded.  The result then follows from \cite{lipp2016variations, NIPS2009}. \frQED
\end{proof}

\section{Numerical Examples}\label{sec:sim}

\subsection{Intuitive Example}
To illustrate the proposed approach, we first consider the system
\begin{align*}
\dot{x}(t)&=\begin{bmatrix} 0 & 1 \\ 0 & 0 \end{bmatrix}x(t)+\begin{bmatrix}1 \\ 1 \end{bmatrix}u(t),\quad y(t)=\begin{bmatrix} 1 & 0 \end{bmatrix}x(t),
\end{align*}
which is observable from $C=[1\ 0]$ but, by inspection, becomes unobservable under $u=-x_2$, as this decouples $x_1$ from $x_2$. We seek a controller that reduces observability from $C$ while remaining within a hard bound on its distance from optimality.
\begin{figure}[!t] 
		\centering
                \includegraphics[width=8cm,height=3.5cm]{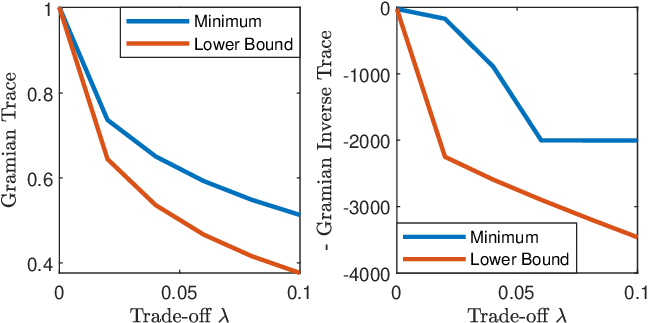}
                \caption{\small Minimum values of Problems \ref{pr:tr} and \ref{pr:trinv} as computed by Algorithms \ref{al:SDP} and \ref{al:SSPinv}, together with the lower bounds from Section \ref{sec:tradeoff}, for various values of $\lambda$.}
                \label{fig:tradeoff}
\end{figure}
We run Algorithm \ref{al:SSPinv} with $Q=0.2I_2$, $R=1$, $\delta=10^{-3}$, $\epsilon=10^{-4}$, $V=I_2$, and $K_0=K^\star$ from the ARE \eqref{eq:realARE}, solving the SDPs with CVX \cite{gb08}. With $\lambda=0.01$ and $\lambda=0.1$, the closed-loop matrices $A+B\hat{K}$ that we obtain based on the output of Algorithm \ref{al:SSPinv} are
\begin{align*}
\lambda=0.01:~~A+B\hat{K}&=\begin{bmatrix} -0.4873  &  0.1905 \\ -0.4873  & -0.8095 \end{bmatrix},\\\lambda=0.1:~~ A+B\hat{K}&=\begin{bmatrix} -0.8572  & -0.0018 \\   -0.8572  & -1.0018 \end{bmatrix},
\end{align*}
whereas the nominal optimal closed-loop matrix $A+BK^\star$ is:
\begin{align*}
A+BK^\star=\begin{bmatrix} -0.4472 &   0.3095 \\
   -0.4472  & -0.6905 \end{bmatrix}.
\end{align*}
As expected, greater flexibility to deviate from optimality lets $\hat K$ increasingly decouple $x_1$ from $x_2$, driving the upper-right entry of $A+B\hat K$ toward zero, at near-optimal performance.

Figure \ref{fig:tradeoff} shows the minimum values of Problems \ref{pr:tr} and \ref{pr:trinv} computed by Algorithms~\ref{al:SDP} and \ref{al:SSPinv}, together with the lower bounds derived in Section~\ref{sec:tradeoff}, for various values of $\lambda$. We notice the bounds are relatively tight for Problem \ref{pr:tr}, but looser for Problem \ref{pr:trinv} because the observability Gramian becomes increasingly ill-conditioned as $\lambda$ increases and the algorithm computes only a local minimum. The figure also illustrates the intuition behind Theorem \ref{th:lb}: minimizing $\mathrm{tr}(W)$ yields larger gains at small performance trade-offs, whereas maximizing $\mathrm{tr}(W^{-1})$ becomes increasingly effective as the allowable performance degradation grows.

\subsection{ADMIRE Aircraft}

We perform further simulations on the ADMIRE aircraft \cite{admire}, which has $n=5$ states and $m=7$ control inputs. We assume that an adversary can measure the first, third, and fifth state, while we refer to \cite{admire} for the expressions of the matrices $A, B$.

\begin{figure}[t]
    \centering
    
    \begin{subfigure}{\linewidth}
        \centering
        \includegraphics[width=0.85\linewidth, height=2.2cm]{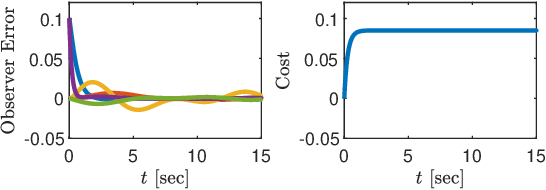}
        \caption{Trajectories under the optimal feedback gain.}
        \label{fig:sub1}
    \end{subfigure}

    \vspace{1em}  % Adjust vertical spacing as needed

    \begin{subfigure}{\linewidth}
        \centering
        \includegraphics[width=0.85\linewidth, height=2.2cm]{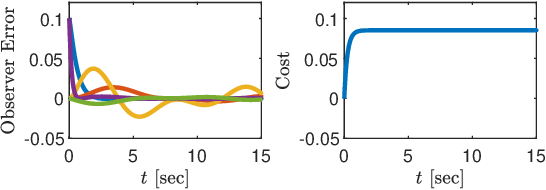}
        \caption{Trajectories under the feedback gain from Algorithm \ref{al:SDP}.}
        \label{fig:sub2}
    \end{subfigure}

    \vspace{1em}

    \begin{subfigure}{\linewidth}
         \centering ~~~
        \includegraphics[width=0.8\linewidth, height=2.2cm]{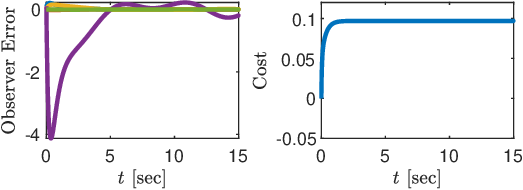}
        \caption{Trajectories under the feedback gain from Algorithm \ref{al:SSPinv}.}
        \label{fig:sub3}
    \end{subfigure}

    \caption{Trajectories of the adversary's observation errors, and the closed-loop performance cost.}
    \label{fig:three-stack}
\end{figure}

\begin{table}[!t]
  \begin{center}
  \caption{\small Eigenvalues under the nominal optimal gain $K^\star$ and the gain obtained by Algorithms \ref{al:SDP}-\ref{al:SSPinv}, and the corresponding costs. \label{tab:ratio2}}
  \resizebox{\columnwidth}{!}{%
   \begin{tabular}{c||c|c|c|c|c|c|c}
      \toprule % <-- Toprule here
      \textbf{ } &  $\textrm{tr}(W)$ &  $\textrm{tr}(W^{-1})$ &  $\lambda_1$ &  $\lambda_2$ &  $\lambda_3$ &  $\lambda_4$ &  $\lambda_5$\\
      \midrule % <-- Midrule here
      Nom.  &  $13.82$ &  $558$ &   $9.79$ &  $3.63$ &  $0.341$ &  $0.0709$ &  $0.00180$\\
       Alg. 1 &  $11.01$ &  $598$ &   $9.34$ &  $1.28$ &  $0.331$ & $ 0.0520$ &  $0.00180$  \\
        Alg. 2 &  $13.80$ &  $31504$ &   $9.78$ &  $3.60$ &  $0.351$ &  $0.0707$ &  $0.00003$  \\
      \bottomrule % <-- Bottomrule here
    \end{tabular}}
  \end{center}
\end{table}

First, we run Algorithm \ref{al:SSPinv} to obtain a gain $\hat{K}$ balancing observability from $C$ against closed-loop performance, with $Q=I_5$, $R=10I_7$, $\delta=10$ (scaled to this setup's cost), $\epsilon=10^{-5}$, $V=I_5$, $\lambda=1$, and $K_0=K^\star$ from the ARE \eqref{eq:realARE}. Table~\ref{tab:ratio2} reports the resulting eigenvalues and trace of the inverse observability Gramian of $(A+B\hat K,C)$. We observe that the minimum eigenvalue is driven very close to zero, as expected. This occurs because maximizing $\mathrm{tr}(W^{-1})$ drives the Gramian toward singularity, most easily through its smallest eigenvalue. Consequently, $(A+B\hat{K},C)$ is nearly unobservable, so an adversary knowing $\hat{K}$ and observing $y(t)$ cannot easily reconstruct $x(t)$. At the same time, closed-loop performance remains near-optimal, with $\mathrm{tr}(W^{-1})$ going from $558$ to $31504$ for only $\lambda=1$.

Second, we run Algorithm \ref{al:SDP} with the same parameters. As Table~\ref{tab:ratio2} shows, the pattern differs: rather than the smallest eigenvalue, the algorithm reduces the larger ones, since more cost is recovered there than from the already small minimum eigenvalue. Moreover, only $\lambda_2$ and $\lambda_4$ drop significantly, consistent with the adversary observing three of the five states.

Finally, we simulate an adversary running a Luenberger observer (with poles at $-15,-5,-2,-2,-1$) whose measurements are corrupted by deterministic noise consisting of five sinusoids (frequencies in $[0,\tfrac{1}{2\pi}]$, amplitude $0.01$). Figure~\ref{fig:three-stack} shows the adversary's estimation errors and the closed-loop performance cost under (a) the nominal controller, (b) Algorithm~\ref{al:SDP}, and (c) Algorithm~\ref{al:SSPinv}. Algorithm~\ref{al:SDP} magnifies estimation errors across all states with negligible performance loss, whereas Algorithm~\ref{al:SSPinv} drives the fourth-state estimation error to unreliable levels at a slightly higher performance cost. These results are consistent with the metrics in \eqref{eq:obsfuns}: $\mathrm{tr}(W)$ captures observability across all state directions, whereas $\mathrm{tr}(W^{-1})$ is dominated by the least observable direction, which here corresponds to the fourth state.

\subsection{Satellite Relative Motion}

We finally consider a target satellite maneuvering in a local Clohessy--Wiltshire frame, observed by a satellite seeking to estimate its state. In particular, the observer seeks to determine the target's radial position during the maneuver, which the target wishes to conceal. The target's state is
$x=[x_r,\ x_a,\ \dot{x}_r,\ \dot{x}_a]^\mathrm{T}$, where $x_r$ and $x_a$ denote the radial and along-track coordinates, respectively. The observer trails the target and has nominal state $[0,-d,0,0]^\mathrm{T}$, while measuring only the inter-satellite range. Linearization of the range measurement about this nominal geometry yields the sensing matrix
$C=\begin{bmatrix}0&1&0&0\end{bmatrix},$
so that the along-track position is measured directly, whereas the radial position is observable only through the Coriolis coupling in the Clohessy--Wiltshire dynamics. It is this coupling that Algorithm~\ref{al:SSPinv} suppresses.

\begin{figure}[!t]
    \centering
    \includegraphics[width=\linewidth]{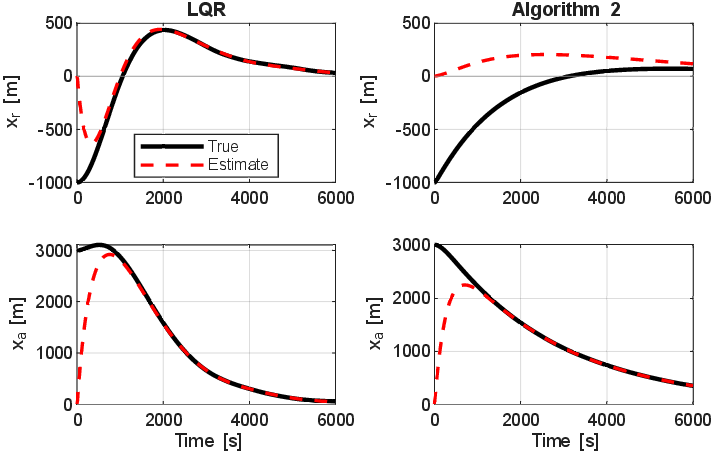}
    \caption{\small  Kalman--Bucy reconstruction of the target state from noisy range measurements under the nominal linear-quadratic controller and the controller synthesized by Algorithm~\ref{al:SSPinv}. While the along-track position $x_a$ is accurately reconstructed in both cases, the radial position $x_r$ exhibits large estimation errors under the proposed controller.}
    \label{fig:KF_reconstruction}
\end{figure}

We run Algorithm~\ref{al:SSPinv} with $Q=I_4$, $R=100I_2$, $\delta=10$, $\epsilon=10^{-3}$, and $\lambda=50$. The adversary is assumed to know $A$, $B$, $C$, and the implemented feedback gain, and employs a Kalman--Bucy observer on noisy range measurements (10~m standard deviation), with the target starting at $-1$~km radial and $3$~km along-track offset. Figure~\ref{fig:KF_reconstruction} shows that the along-track position is accurately reconstructed under both controllers, whereas the radial position is accurately reconstructed only under the nominal controller and exhibits large estimation errors under the controller synthesized by Algorithm~\ref{al:SSPinv}. These results demonstrate that reducing adversarial observability directly degrades state reconstruction in a physically motivated sensing scenario.

%\balance
\bibliographystyle{ieeetr}      
\bibliography{references.bib}

\end{document}